\documentclass{llncs}
\usepackage[OT2,OT1]{fontenc}
\newcommand\cyr
{
\renewcommand\rmdefault{wncyr}
\renewcommand\sfdefault{wncyss}
\renewcommand\encodingdefault{OT2}
\normalfont
\selectfont
}
\DeclareTextFontCommand{\textcyr}{\cyr}
\def\cprime{\char"7E }

\usepackage{latexsym}
\usepackage{amsmath,amssymb,amscd}
\usepackage{bussproofs}
\usepackage{graphicx}
\graphicspath{ {images/} }
\usepackage{xspace}
\usepackage{amsmath,amssymb,amscd,amsfonts,latexsym}
\usepackage{bussproofs}
\usepackage{color}
\usepackage{multicol}
\usepackage[english]{babel}
\usepackage{hyperref}
\usepackage{tikz}
\usetikzlibrary{matrix}
\usetikzlibrary{positioning}
\usepackage{yfonts}
\usepackage{natded}
\usepackage{proof}

\newcommand{\cblue}[1]{\textcolor{blue}{#1}}

\newcommand{\CL}{\mathsf{CL}}
\newcommand{\CPL}{\mathsf{CPL}}
\newcommand{\IL}{\mathsf{IL}}
\newcommand{\IPL}{\mathsf{IPL}}
\newcommand{\ML}{\mathsf{ML}}
\newcommand{\PA}{\mathsf{PA}}
\newcommand{\HA}{\mathsf{HA}}
\newcommand{\NE}{\mathsf{NE}}

\newcommand{\NK}{\mathsf{NK}}
\newcommand{\NEK}{\mathsf{NE_{K}}}

\newcommand{\ECI}{\mathsf{ECI}}

\date{}
\title{Glivenko's theorems from an ecumenical perspective}
\author{Luiz Carlos Pereira, Victor Barroso-Nascimento and Elaine Pimentel}
\institute{}

\begin{document}
\maketitle
\setcounter{secnumdepth}{3}

\begin{abstract}
In this paper, we revisit Glivenko's theorems, foundational results relating classical and intuitionistic logic, from an ecumenical perspective. We begin by discussing the historical context and significance of Glivenko's original contributions, and then examine their extensions and reinterpretations within ecumenical logical frameworks. Our analysis focuses on three ecumenical systems: Prawitz's natural deduction system $\NE$; the system $\NEK$, closely related to one introduced by Krauss in an unpublished manuscript; and the $\ECI$ system proposed by Barroso-Nascimento.
\end{abstract}

\section{Introduction}
In the late twenties and early thirties of last century, several results were obtained concerning some relations between classical logic ($\CL$) and intuitionistic logic ($\IL$), as well as between classical arithmetic ($\PA$) and intutionistic arithmetic ($\HA$). In 1925, Kolmogorov proved that classical propositional logic ($\CPL$) could be translated into intuitionistic propositional logic ($\IPL$)~\cite{Kolmogorov}. In 1933, G\"{o}del defined an interpretation of $\PA$ into $\HA$~\cite{Godel}  and in the same year Gentzen defined a different interpretation of $\PA$ into $\HA$~\cite{gentzen1969}. These interpretations/translations\footnote{See~\cite{Gilda-Oliva1,Gilda-Oliva2} for illuminating presentations and discussions of these translations, as well as their relation to two other translations due to Kuroda and Krivine \cite{Kur51,Kriv90}.} were defined as functions from the language of $\PA$ (or $\CL$, $\CPL$) into some fragment of the language of $\HA$ ($\IL$, $\IPL$) that aimed to preserve some important properties, like theoremhood or derivability.   What is known as \textit{Glivenko's theorems} in the area of logic belongs to this group of important results.

\noindent Valery Glivenko's  results were published in 1929, in French, in the \textit{Bulletins de la Classe des Sciences de la Acad\'{e}mie Royale de Belgique}, under the title \textit{Sur quelques points de la  logique de M. Brouwer}~\cite{Glivenko29}. The first Glivenko theorem establishes that if a formula $A$ is classically provable in $\CPL$, then its double negation is intuitionistically provable in $\IPL$\footnote{It is both interesting and important to observe (although it has not been frequently noted!) that when an intuitionistic logician proves, for example, $\neg \neg (A \vee B)$ while a classical logician proves $(A \vee B)$, the connective $\vee$ in $\neg \neg (A \vee B)$ does not carry the same meaning as the connective $\vee$ in $(A \vee B)$. The proof of $\neg \neg (A \vee B)$ belongs to the intuitionistic system and, although we use the same symbol $\vee$, within that system it has an intuitionistic interpretation rather than a classical one. The same observation applies to the familiar double-negation translations: the operators in the \textit{image-language} inherit their meaning from the semantics of the \textit{image-language}. This general point motivates the practice of distinguishing between classical and intuitionistic operators by means of different symbols, as is done in ecumenical systems. For instance, in Prawitz's ecumenical system one finds both a classical disjunction $\vee_{c}$ and an intuitionistic disjunction $\vee_{i}$.}.

\begin{theorem} 
If $\vdash_{\CPL} A$, then $\vdash_{\IPL} \neg \neg A$.
\end{theorem}

In all of the logics considered in this paper, $\neg A$ is defined as $A \to \bot$, a convention we may also adopt in classical and intuitionistic logic.


This theorem is known to hold in full generality only for propositional logic. However, an immediate corollary of Seldin's normalization strategy for first-order classical logic\footnote{For an insightful presentation and discussion of Seldin's normalization strategy, see~\cite{GiuNaibo19}.} \cite{Sel89,PHCS10}, together with its translation into intuitionistic first-order logic due to Kuroda~\cite{Kur51,Gilda-Oliva1}, is that the theorem also holds for first-order formulas that do not contain universal quantifiers. Moreover, Andr\'{e}s Raggio derived the normalization theorem for Gentzen's classical Natural Deduction system $\NK$ as a consequence of Glivenko's first theorem~\cite{Raggio65}. This shows that the theorem is closely tied both to translation techniques and to normalization strategies.

The second Glivenko theorem establishes that if a formula $\neg A$ is provable $\CPL$, then the same formula  is provable in $\IPL$:

\begin{theorem} 
If $\vdash_{\CPL} \neg A$, then $\vdash_{\IPL} \neg A$.
\end{theorem}

The second theorem is a trivial consequence of the first: by the first theorem, $\vdash_{\CPL} \neg A$ implies $\vdash_{\IPL} \neg \neg \neg A$, and $\vdash_{\IPL} \neg \neg \neg A$ intuitionistically implies $\vdash_{\IPL} \neg A$.  

In a way, Glivenko's theorems allow classical validities to be sought {\em constructively}.  This allows us to conceive propositional classical logic as a part of propositional intuitionistic logic, the latter being capable of making more {\em fine-grained} distinctions than the former.


This paper examines Glivenko's theorems through the lens of ecumenical logic, focusing on their implications and extensions within a unified logical framework. We begin by revisiting Glivenko's original results and their historical context, emphasizing their significance in bridging the gap between classical and intuitionistic logic. Building on this idea, we explore the application of ecumenical systems, such as those proposed by Prawitz, Krauss, and Barroso-Nascimento, to formalize and generalize Glivenko-type results. 
Finally, we argue that the ecumenical perspective sheds light on the interplay between classical and intuitionistic {\em reasoning}, offering a deeper understanding of their coexistence within a single system while respecting their distinct inferential principles. 




\section{Glivenko-type results in Prawitz's ecumenical system}

In 2015, Dag Prawitz proposed a natural deduction system where classical logic and intuitionistic logic could both be codified~\cite{DBLP:journals/Prawitz15}. Prawitz's system is an example of what nowadays is called an \textit{ecumenical system}~\cite{PP23}. Ecumenical systems allow two or more logics, even rival ones, to {\em coexist peacefully}. This peaceful coexistence means that the combination will not produce a collapse of a weaker logic into a stronger one, thus naturally preserving the essential characteristics of the logics involved in the combination\footnote{An abstract study of non-collapsing combinations of logics can be found in~\cite{RS25}.}.

In Prawitz's system, the classical logician and the intuitionistic logician would share the universal quantifier, conjunction, negation and the constant for the absurd (the so-called \textit{neutral operators}), but they would each have their own existential quantifier, disjunction and implication, with different meanings ($\exists_j,\vee_j,\to_j$, where $j\in\{i,c\}$ for the intuitionistic and classical versions, respectively). Also, classical and intuitionistic  $n$-ary predicate letters ($P_{c}^{n}, P_{i}^{n},\ldots$) co-exist  but have different meanings. Prawitz's main idea is that these different meanings are given by a semantical framework that can be accepted by both parties\footnote{From a semantic perspective, this can be done either by defining clauses for operators of stronger logics in a semantic framework for the weaker logic~\cite{PR17} or by also defining distinct semantic notions that co-exist peacefully~\cite{Barroso-Nascimento25}.}. Prawitz's ecumenical system, here called $\NE$, is shown in Fig.~\ref{fig:NE}. 

\begin{figure}[htp]
{\sc Intuitionistic rules}
\[
\begin{array}{lc@{\qquad}l}
\infer[{\to_{i}\mbox{-elim}}]{B}{A\to_{i} B & A}
&
\infer[{\to_{i}\mbox{-int}}]{ A\to_{i} B}
{\deduce{B}{\deduce{}{\deduce{\Pi}{\deduce{}{[A]}}}}} 
&
\infer[{\vee_i\mbox{-elim}}]{C}{ A\vee_i B & \deduce{C}{\deduce{}{\deduce{\Pi_1}{\deduce{}{[A]}}}}
 &\deduce{C}{\deduce{}{\deduce{\Pi_2}{\deduce{}{[B]}}}}} 
\\[5pt]
\infer[{\vee_i\mbox{-int}_j}]{ A_1\vee_i A_2}{ A_j}
&
\infer[\exists_i\mbox{-elim}]{ B}{ \exists_ix.A &  \deduce{B}{\deduce{}{\deduce{\Pi}{\deduce{}{[A(a/x)]}}}}}
&
\infer[\exists_i\mbox{-int}]{\exists_ix.A }{ A(t/x)}
\end{array}
\]
{\sc Classical rules}
\[
\begin{array}{lc@{\qquad}l}
\infer[{\to_c\mbox{-elim}}]{\bot}{ A\to_c B & A & \neg B}
&
\infer[{\to_c\mbox{-int}}]{A\to_c B}
{\deduce{\bot}{\deduce{}{\deduce{\Pi}{\deduce{}{[A,\neg B]}}}}} 
&
\infer[{\vee_c\mbox{-elim}}]{\bot}{A\vee_c B & \neg A &  \neg B} 
\\[5pt]
\infer[{\vee_c\mbox{-int}}]{A\vee_c B}{\deduce{\bot}{\deduce{}{\deduce{\Pi}{\deduce{}{[\neg A,\neg B]}}}}} 
&
\infer[\exists_c\mbox{-elim}]{\bot}{ \exists_c x A &  \forall_{i} x \neg A}
&
\infer[\exists_c\mbox{-int}]{ \exists_c x A}{\deduce{\bot}{\deduce{}{\deduce{\Pi}{\deduce{}{[\forall_{i} x \neg A]}}}}}
\\[5pt]
\infer[P^n_c\mbox{-elim}]{ \bot}
{
P^n_c(t_{1},...,t_{n}) &  \neg P^n_i(t_{1},...,t_{n})
}
&
\infer[P^n_c\mbox{-int}]{ P^n_c(t_{1},...,t_{n})}
{
\deduce{\bot}{\deduce{}{\deduce{\Pi}{\deduce{}{[\neg P^n_i(t_{1},...,t_{n})}]}}}
}
\end{array}
\]
{\sc Neutral rules}
\[
\begin{array}{lc@{\qquad}lc@{\qquad}lc@{\qquad}l}
\infer[{\wedge\mbox{-elim}_j}]{ A_j}{A_1\wedge A_2} 
& &
\infer[{\wedge\mbox{-int}}]{A \wedge B}{A \quad  B} 
& &
\infer[\bot\mbox{-elim}]{A}{ \bot}
& &\\[5pt]
\infer[\forall\mbox{-elim}]{A(t/x) }{\forall x.A}
& &
\infer[\forall\mbox{-int}]{ \forall x.A}{ A(a/x)}
\end{array}
\]

\caption{Ecumenical natural deduction system $\NE$.  In rules 
$\forall\mbox{-int}$ and $\exists_i\mbox{-elim}$, the parameter $a$ is fresh. In $\exists_i\mbox{-int}$ and $\forall\mbox{-elim}$, $t$ is a term.}\label{fig:NE}
\end{figure}

It is obvious that we cannot have \textit{Glivenko's theorems} in Prawitz's ecumenical system for the plain reason that we do not have \textit{two systems}, the intuitionistic system and the classical system, but only one, \textit{the ecumenical system}. However, we can have a kind of \textit{internal Glivenko}, which establishes Glivenko-type relations between classical operators and intuitionistic operators.
\begin{theorem}
For any formula $C$ such that the main operator of $C$ is a classical operator, we have that $C \vdash_{\NE} \neg \neg C^{*}$, where $C^{*}$ is the result of replacing the classical operator by the corresponding intutionistic operator.
\end{theorem}
\begin{proof}
We examine below the case of each classical operator:
\begin{enumerate}
\item $C$ is $A \vee_{c} B$. We can prove that $A \vee_{c} B \vdash_{\NE}  \neg \neg (A \vee_{i} B)$ as follows:

\begin{prooftree}
\AxiomC{$A \vee_{c} B$}
\AxiomC{$[A]^{1}$}
\UnaryInfC{$A \vee_{i} B$}
\AxiomC{$[\neg (A \vee_{i} B)]^{3}$}
\BinaryInfC{$\bot$}
\RightLabel{1}
\UnaryInfC{$\neg A$}
\AxiomC{$[B]^{2}$}
\UnaryInfC{$A \vee_{i} B$}
\AxiomC{$[\neg (A \vee_{i} B)]^{3}$}
\BinaryInfC{$\bot$}
\RightLabel{2}
\UnaryInfC{$\neg B$}
\TrinaryInfC{$\bot$}
\RightLabel{3}
\UnaryInfC{$\neg \neg (A \vee_{i} B)$}
\end{prooftree}

\item $C$ is $A \to_{c} B$. We can prove that $A \vee B \vdash_{\NE}  \neg \neg (A \to_{i} B)$ as follows:

\begin{prooftree}
\AxiomC{$A \to_{c} B$}
\AxiomC{$[A]^{2}$}
\AxiomC{$[B]^{1}$}
\UnaryInfC{$A \to_{i} B$}
\AxiomC{$[\neg (A \to_{i} B)]^{3}$}
\BinaryInfC{$\bot$}
\RightLabel{1}
\UnaryInfC{$\neg B$}
\TrinaryInfC{$\bot$}
\UnaryInfC{$B$}
\RightLabel{2}
\UnaryInfC{$A \to_{i} B$}
\AxiomC{$[\neg (A \to_{i} B)]^{3}$}
\BinaryInfC{$\bot$}
\RightLabel{3}
\UnaryInfC{$\neg \neg (A \to_{i} B)$}
\end{prooftree}

\item $C$ is $\exists_{c} xA(x)$. We can prove that $\exists_{c} xA(x) \vdash_{\NE}  \neg \neg \exists_{i} xA(x)$ as follows:

\begin{prooftree}
\AxiomC{$\exists_c xA(x)$}
\AxiomC{$[A(a/x)]^{1}$}
\UnaryInfC{$\exists_i xA(x)$}
\AxiomC{$[\neg (\exists_i xA(x))]^{2}$}
\BinaryInfC{$\bot$}
\RightLabel{1}
\UnaryInfC{$\neg A(a/x)$}
\UnaryInfC{$\forall x \neg A(x)$}
\BinaryInfC{$\bot$}
\RightLabel{2}
\UnaryInfC{$\neg \neg (\exists_i xA(x))$}
\end{prooftree}

\end{enumerate}
\end{proof}

\vspace{1.0cm}
\noindent We can also have an \textit{internal} result  corresponding to Glivenko's second theorem:
\begin{theorem}
For any formula $C$ such that the main operator of $C$ is a classical operator, we have that $\neg C \vdash_{\NE} \neg C^{*}$, where $C^{*}$ is the result of replacing the classical operator by the corresponding intutionistic operator.
\end{theorem}
\begin{proof} The result follows directly from the fact that $C^{*} \vdash_{\NE} C$.\\
For example:
\begin{prooftree}
\AxiomC{$A \vee_{i} B$}
\AxiomC{$[A]^{1}$}
\AxiomC{$[\neg A]^{3}$}
\BinaryInfC{$\bot$}
\AxiomC{$[B]^{2}$}
\AxiomC{$[\neg B]^{4}$}
\BinaryInfC{$\bot$}
\RightLabel{1,2}
\TrinaryInfC{$\bot$}
\RightLabel{3,4}
\UnaryInfC{$A \vee_{c} B$}
\end{prooftree}
\end{proof}
\section{Glivenko-type results in the ecumenical system $\ECI$}
Another ecumenical system, the system  $\ECI$, was introduced by Victor Barroso-Nascimento, and its main idea is:

\begin{quote}
   [$\ldots$] the generalist approach consists of adding general rules which allow the introduction of ecumenical versions of any formula. Thus, instead of directly defining assertion conditions for specific ecumenical connectives, the generalist approach aims to define assertion conditions for ecumenical formulas in general, so as we can introduce the remaining ecumenical operators as special cases of the general rule ~\cite[pg. 38, translated from the original in Portuguese]{Barroso-Nascimento}).
\end{quote}

Instead of the ecumenical operators of Prawitz's ``inferentialist" approach\footnote{For reasons discussed in the last section of this paper and pointed out by Luiz Carlos Pereira in other contexts, using ``inferentialist" and ``generalist" to refer to the two approaches is rather misleading.}, we have ecumenical {\em formulas}, $A$ 
and $A^{c}$\footnote{For easing the notation, we will omit the superscript of the intuitionistic formulas, marking only the classical ones.}. The system $\ECI$ is obtained from Prawitz's natural deduction system for intuitionistic logic by adding the following rules:

\begin{prooftree}
\AxiomC{$[\neg A]$}
\noLine
\UnaryInfC{$\Pi$}
\noLine
\UnaryInfC{$\bot$}
\RightLabel{$I_{C}$}
\UnaryInfC{$A^{c}$}
\DisplayProof
\qquad
\AxiomC{$A^{c}$}
\AxiomC{$\neg A$}
\RightLabel{$E_{C}$}
\BinaryInfC{$\bot$}
\end{prooftree}

\noindent We can immediately see that \textit{Glivenko-type} results are somehow trivial in $\ECI$\footnote{The system $\ECI$ is defined over a language that does not have any explicit sign that would \textit{mark} a formula as being \textit{intuitionistic}. It is implicit that if the formula is not \textit{labeled} with the sign/constant $c$, then it is an intuitionistic formula. A formula without any occurrence of the label $c$ is a \textit{full intuitionistic formula}.}:
\begin{theorem}\label{thm:G1}
$A^{c} \vdash_{\ECI} \neg \neg A$
\end{theorem}
\begin{proof}
Immediate by the following derivation 
\begin{prooftree}
\AxiomC{$A^{c}$}
\AxiomC{$[\neg A]^{1}$}
\BinaryInfC{$\bot$}
\RightLabel{1}
\UnaryInfC{$\neg \neg A$}
\end{prooftree}

\end{proof}

%

\begin{theorem}\label{thm:G2}
$(\neg A)^{c} \vdash_{\ECI} \neg A$\footnote{It is interesting to observe that $(\neg A)^{c} \dashv \vdash_{\ECI} \neg (A^{c})$.}
\end{theorem}
\begin{proof} Direct from Theorem~\ref{thm:G1} and by the fact that $\neg \neg \neg A \vdash_{\ECI} \neg A$.
\end{proof}

\subsection{The translation $t\ECI$}
We have shown that $A^{c} \vdash_{\ECI} \neg \neg A$, but we can trivially also show that\\ $\neg \neg A \vdash_{\ECI} A^{c}$:
\begin{prooftree}
\AxiomC{$[\neg A]^{1}$}
\AxiomC{$\neg \neg A$}
\BinaryInfC{$\bot$}
\RightLabel{1}
\UnaryInfC{$A^{c}$}
\end{prooftree}

%

\noindent In a certain sense, the equivalence between $\neg \neg A$ and $A^{c}$ provides a justification for the translation presented on page~52 of~\cite{Barroso-Nascimento}, here reformulated in a recursive manner.

\begin{definition}\label{TECI}
 $t\ECI[A]$ is defined as follows, where A is a formula of $\ECI$:

    \begin{enumerate}
        \item $t\ECI[p] = p$, for atomic $p$;
        \smallskip
        \item $t\ECI[\bot] = \bot$;
        \smallskip
        \item $t\ECI[A \star B] = t\ECI[A] \star t\ECI[B]$, for $\star \in \{\land, \lor, \to\};$
        \smallskip
        \item $t\ECI[A^{c}] = \neg \neg \  t\ECI[A]$.
    \end{enumerate}
\end{definition}

This translation is used to reduce the problem of \emph{normalization} (and other proof-theoretical results) for 
$\ECI$ to normalization in intuitionistic logic. However, if we are interested in \emph{identifying} the effects of classical reasoning within a derivation, there is another aspect of this translation that deserves attention.

Assume that we have a derivation 
in classical propositional logic with several applications of the classical \textit{reductio}.
\begin{prooftree}
\AxiomC{$[\neg A]$}
\noLine
\UnaryInfC{$\Pi$}
\noLine
\UnaryInfC{$\bot$}
\UnaryInfC{$A$}
\end{prooftree}
Suppose now that we replace each such application by an application of $\to$-int:
\begin{prooftree}
\AxiomC{$[\neg A]^1$}
\noLine
\UnaryInfC{$\Pi$}
\noLine
\UnaryInfC{$\bot$}
\RightLabel{1}
\UnaryInfC{$\neg \neg A$}
\end{prooftree}
It is easy to see that the resulting derivation need not be intuitionistically valid. Consider, for instance, the following  derivation: 
\begin{prooftree}
\AxiomC{$[\neg A]$}
\noLine
\UnaryInfC{$\Pi$}
\noLine
\UnaryInfC{$\bot$}
\UnaryInfC{$A$}
\AxiomC{$\Pi'$}
\noLine
\UnaryInfC{$A \to B$}
\BinaryInfC{$B$}
\end{prooftree}
Under the above transformation, this derivation becomes
\begin{prooftree}
\AxiomC{$[\neg A]$}
\noLine
\UnaryInfC{$\Pi$}
\noLine
\UnaryInfC{$\bot$}
\UnaryInfC{$\neg \neg A$}
\AxiomC{$\Pi'$}
\noLine
\UnaryInfC{$A \to B$}
\BinaryInfC{$B$}
\end{prooftree}

\noindent which is not a legitimate intuitionistic derivation.

The same problem would  arise if we take into consideration the translation $t\ECI$. The derivation would be be transformed into:
\begin{prooftree}
\AxiomC{$[\neg A]$}
\noLine
\UnaryInfC{$\Pi$}
\noLine
\UnaryInfC{$\bot$}
\UnaryInfC{$A^{c}$}
\AxiomC{$\Pi'$}
\noLine
\UnaryInfC{$A \to B$}
\BinaryInfC{$B$}
\end{prooftree}

\noindent But this derivation is also not a legitimate derivation. However, we can replace this derivation by:

\begin{prooftree}
\AxiomC{$[A]$}
\AxiomC{$A\to B$}
\BinaryInfC{$B$}
\AxiomC{$[\neg B]$}
\BinaryInfC{$\bot$}
\UnaryInfC{$ \neg A$}
\noLine
\UnaryInfC{$\Pi$}
\noLine
\UnaryInfC{$\bot$}
\UnaryInfC{$B^{c}$}
\end{prooftree}


\noindent And now we can see that the formula $B$ that is derived depending on an application of classical reasoning has a \textit{classical nature} too.
This shows that the system $\ECI$ can, in a precise way, support Krauss' insight~\cite{Krauss} that an ecumenical perspective helps us identify where classical reasoning is actually needed (and, importantly, that we need not be classical everywhere, nor all the time). It also allows us to make explicit the consequences of invoking classical principles within a given derivation. We turn to this point next. 

\section{The system $\NEK$ and the curious case of classical conjunction}

The first ecumenical system for $\CL$ and $\IL$ was proposed and studied in 1992 by Peter Krauss~\cite{Krauss}, although he did not use the terminology ``ecumenical''\footnote{Krauss obtains this system by first defining an ecumenical system for $\CL$ and minimal logic ($\ML$), then extending it to a system for $\IL$ and $\CL$ by adding the rule $\bot$-elim. He proceeds to prove results for both ecumenical systems. Two ecumenical systems containing rules for $\IL$ and $\ML$ are presented by Barroso-Nascimento in \cite{Barroso-Nascimento}, one defining rules for operators (as in $\NE$) and one defining rules for formulas (as in $\ECI$).}. In addition to three rules for equality, a classical disjunction, a classical implication, and a classical existential quantifier, Krauss' system also has a classical conjunction $\wedge_{c}$\footnote{Although he accepts the existence of two conjunctions, Krauss claims that the classical mathematician seldom uses the classical one. In fact, he observes that the classical conjunction is not idempotent: from a proof of $A \wedge_c A$ one can derive only $\neg\neg A$, which always entails $A$ in classical logic, but not in the ecumenical setting.} and a classical universal quantifier $\forall_{c}$, retaining only negation and $\bot$ as a neutral operator\footnote{We observe that the idea of having two conjunctions appears in several works, such as Girard's Constructive Classical Logic~\cite{Girard91}, Liang and Miller's focused systems~\cite{LiaMil09}, as well as in ecumenical approaches to automated deduction~\cite{Emilie}.}.

In what follows we define a system $\NEK$, which is essentially the same system presented in \cite{PPDP25} (the only difference being the inclusion of classical atoms) and can be proven to be equivalent to Krauss' original system (modulo inclusion of equality and removal of classical atoms). This system is obtained by adding the rules in Fig. \ref{fig:NEK} to Prawitz's $\NE$.

\begin{figure}[t]
\[
\begin{array}{lc@{\qquad}l}
\infer[{\land_c\mbox{-elim}_j}]{\bot}{A_1 \land_{c} A_2 & \neg A_j}
& &
\infer[{\land_c\mbox{-int}}]{A \land_{c} B}{\deduce{\bot}{\deduce{}{\deduce{\Pi_1}{\deduce{}{[ \neg A]}}}}
 &\deduce{\bot}{\deduce{}{\deduce{\Pi_2}{\deduce{}{[\neg B]}}}}} 
\\
\bigskip
\\
\infer[\forall_c\mbox{-elim}]{\neg \neg A(t/x)}{ \forall_c x A(x)}
& &
\infer[\forall_c\mbox{-int}]{\forall_c x A(x)}{ \deduce{\bot}{\deduce{}{\deduce{\Pi}{\deduce{}{[\exists_i x \neg A(x)]}}}}}
\end{array}
\]

\caption{Rules added to $\NE$ in order to obtain the system $\NEK$. Since they are no longer neutral, intuitionistic conjunctions in $\NEK$ are also represented by $\land_{i}$ instead of $\land$.}\label{fig:NEK}
\end{figure}

\noindent We can easily show that classical conjunction $\wedge_{c}$ satisfies our \textit{internal Glivenko theorems}.  
\begin{lemma} 
$B \wedge_{c} C\vdash_{\NEK} \neg \neg (B \wedge_{i} C)$
\end{lemma}
\begin{prooftree}
\AxiomC{$B \wedge_{c} C$}
\AxiomC{$B \wedge_{c} C$}
\AxiomC{$[B]^{1}$}
\AxiomC{$[C]^{2}$}
\BinaryInfC{$B \wedge_{i} C$}
\AxiomC{$[\neg (B \wedge_{i} C)]^{3}$}
\BinaryInfC{$\bot$}
\RightLabel{$1$}
\UnaryInfC{$\neg B$}
\BinaryInfC{$\bot$}
\RightLabel{$2$}
\UnaryInfC{$\neg C$}
\BinaryInfC{$\bot$}
\RightLabel{$3$}
\UnaryInfC{$\neg \neg (B \wedge_{i} C)$}
\end{prooftree}

\vspace{0.5cm}

\begin{lemma} 
$B \wedge_{i} C \vdash_{\NEK} B \wedge_{c} C$
\end{lemma}

\begin{prooftree}
\AxiomC{$B \wedge_{i} C$}
\UnaryInfC{$B$}
\AxiomC{$[\neg B]^{1}$}
\BinaryInfC{$\bot$}
\AxiomC{$B \wedge_{i} C$}
\UnaryInfC{$C$}
\AxiomC{$[\neg C]^{2}$}
\BinaryInfC{$\bot$}
\RightLabel{$1,2$}
\BinaryInfC{$B \wedge_{c} C$}
\end{prooftree}
From this lemma we can directly conclude the following:
\begin{corollary}
$\neg (B \wedge_{c} C)\vdash_{\NEK} \neg (B \wedge_{i} C)$
\end{corollary}
It turns out that, when we restrict attention to the propositional fragment of $\ECI$, we can emulate the behavior of the rules $I_{c}$ and $E_{c}$ of $\ECI$ with respect to conjunction by means of applications of $\wedge_{c}$-int and $\wedge_{c}$-elim$_j$ respectively. This will be addressed in the next section.

\section{Deductive equivalence of $\forall$-free $\NEK$ and $\ECI$}

In this section we prove that, in first-order logic without universal quantification, $\NEK$ and $\ECI$ are deductively equivalent. Equivalence results for $\ECI$ and Prawitz's $\NE$ in the common fragment of their languages are established in Theorem 4 of \cite{Barroso-Nascimento}. Since $\NEK$ is obtained from $\NE$ by adding classical conjunction, it therefore suffices to consider the induction steps corresponding to $\land_{c}$. For reasons discussed in the next section, this equivalence does not hold for $\NEK$ and $\ECI$ in the presence of universal quantification.


There are at least two standard approaches to comparing two (or more) ecumenical logics.
In the first approach, one shows that for every natural deduction rule $R$ of $L_1$ whose formulation uses only a fragment of the language shared by $L_1$ and $L_2$, whenever the premises of $R$ are derivable in $L_2$, so is its conclusion (and symmetrically, the same is shown for the rules of $L_2$ in $L_1$). This is the strategy adopted in \cite{Barroso-Nascimento} to establish proof-theoretic equivalence between $\ECI$ and $\NE$.
In the second approach, rather than restricting attention to the shared fragment of the language, one defines a translation mapping each formula of the stronger logic to a formula of the weaker one. This strategy is used, for instance, in \cite{PR17} to obtain certain semantic results.
Since the choice between these approaches is largely a matter of convenience, we adopt the second one here.

\begin{definition}
    $t\NEK$ is defined as follows, where $A$ is a formula and $\Gamma$ a set of formulas of $\ECI$ not containing any universal quantifiers:

    \begin{enumerate}
        \item $t\NEK[p ^{i}] = p ^{i}$, for atomic $p$ and $i \in \{ i, c\}$;
        \smallskip
        \item $t\NEK[\bot] = t\NEK[(\bot)^{c}] =  \bot$;
        \smallskip
        \item $t\NEK[A \star B] = t\NEK[A] \star t\NEK[B]$, for $\star \in \{\land, \lor, \to \};$
        \smallskip
         \item $t\NEK[(A \star B)^{c}] = t\NEK[A] \star_{c} t\NEK[B]$, for $\star \in \{\land, \lor, \to \};$
         \smallskip
         \item $t\NEK[\exists x A] = \exists_{i} x \ t\NEK [A]$;
         \smallskip
         \item $t\NEK[(\exists x A)^{c}] = \exists_{c} x \ t\NEK [A]$;
         \smallskip
       \item $t\NEK [\Gamma] = \{t\NEK [A] \  | \ A \in \Gamma \}$.
    \end{enumerate}
\end{definition}

\begin{theorem}\label{thm:deductequiv}
    $\Gamma \vdash_{\ECI} A$ iff $t\NEK[\Gamma] \vdash t\NEK[A]$.
\end{theorem}

\begin{proof}
By induction on the length of the derivations, in which we consider the last rule applied in the deduction (if any). The bases case is trivial, as are the cases including introduction and elimination rules for $\neg A$, $A \land B$, $A \lor B$, $A \to B$, $\exists x A$ and classical atoms $P^n_c(t_{1},...,t_{n})$. The step for $(\bot)^{c}$ is also trivial (since $\neg \bot$ is a theorem of $\ECI$). This means that we only have to deal with classical operators. The proofs for $(A \to B)^{c}$, $(A \lor B)^{c}$ and $(\exists x A)^{c}$ can be found in \cite[pgs. 39-42 and 93]{Barroso-Nascimento} and are thus omitted -- with the exception of the case of applications of $E_{c}$ with one premise of shape $(A \to_{c} B)$, which we simplify here.

\medskip

$(\Longrightarrow)$ We show that $\Gamma \vdash_{\NE} A$ implies $t\NEK[\Gamma] \vdash_{\NE} t\NEK[A]$.

\medskip

In order to ease the notation, we simply write $A$ instead of $t\NEK[A]$ when dealing with deductions in $\NEK$. This results in an ambiguity in the case of $\bot$, so we explicitly stipulate that occurrences of $\bot$ specifically stand for $t\NEK[\bot]$.

\begin{enumerate}
    \item The derivation ends with an application of $I_{C}$ with conclusion $A \land_{c} B$. Then it has the following shape:

    \begin{prooftree}
        \AxiomC{$[\neg(A \land B)]$}
        \noLine
        \UnaryInfC{$\Pi$}
        \noLine
        \UnaryInfC{$\bot$}
        \RightLabel{$I_{C}$}
        \UnaryInfC{$(A \land B)^{c}$}
    \end{prooftree}

The inductive hypothesis yields a deduction $\Pi^{*}$ of $\bot$ possibly depending on $\neg(A \land_{i} B)$. We can construct the following derivation of $A \land_{c} B$ in $\NEK$:

\begin{prooftree}
\AxiomC{$[A \wedge_{i} B]^{1}$}
\UnaryInfC{$A$}
\AxiomC{$[\neg A]^{3}$}
\BinaryInfC{$\bot$}
\RightLabel{1}
\UnaryInfC{$[\neg (A \wedge_{i} B)]$}
\noLine
\UnaryInfC{$\Pi^{*}$}
\noLine
\UnaryInfC{$\bot$}
\AxiomC{$[A \wedge_{i} B]^{1}$}
\UnaryInfC{$B$}
\AxiomC{$[\neg  B]^{4}$}
\BinaryInfC{$\bot$}
\RightLabel{1}
\UnaryInfC{$[\neg (A \wedge_{i} B)]$}
\noLine
\UnaryInfC{$\Pi^{*}$}
\noLine
\UnaryInfC{$\bot$}
\RightLabel{3,4}
\BinaryInfC{$A \wedge_{c} B$}
\end{prooftree}

\item The derivation ends with an application of $E_{C}$ which has one premise of shape $(A \land B)^{c}$. Then it has the following shape:

\begin{prooftree}
   \AxiomC{$\Pi_{1}$}
   \noLine
   \UnaryInfC{$(A \land B)^{c}$}
   \AxiomC{$\Pi_{2}$}
   \noLine
   \UnaryInfC{$\neg(A \land_{i} B)$}
   \RightLabel{$E_{C}$}
   \BinaryInfC{$\bot$}
\end{prooftree}

The inductive hypothesis yields a deduction $\Pi^{*}_{1}$ of $A \land_{c} B$ and a deduction $\Pi^{*}_{2}$ of $\neg(A \land_{i} B)$. We can construct the following derivation of $\bot$ in $\NEK$:

\begin{prooftree}
\AxiomC{$\Pi^{*}_{1}$}
\noLine
\UnaryInfC{$A \wedge_{c} B$}
\AxiomC{$\Pi^{*}_{1}$}
\noLine
\UnaryInfC{$A \wedge_{c} B$}
\AxiomC{$[A]^{1}$}
\AxiomC{$[B]^{2}$}
\BinaryInfC{$A \wedge_{i} B$}
\AxiomC{$\Pi^{*}_{2}$}
\noLine
\UnaryInfC{$\neg (A \wedge_{i} B)$}
\BinaryInfC{$\bot$}
\RightLabel{1}
\UnaryInfC{$\neg A$}
\BinaryInfC{$\bot$}
\RightLabel{2}
\UnaryInfC{$\neg B$}
\BinaryInfC{$\bot$}
\end{prooftree}
\noindent

\item The derivation ends with an application of $E_{C}$ with a premise $(A \to B)^{c}$. Then we do the following:

\begin{prooftree}
\AxiomC{$\Pi^{*}_{1}$}
\UnaryInfC{$A \to_{c} B$}
\AxiomC{$[A]^{2}$}
    \AxiomC{$[B]^{1}$}
    \UnaryInfC{$A \to_{i} B$}
    \noLine
    \AxiomC{$\Pi^{*}_{2}$}
    \UnaryInfC{$\neg (A \to _{i} B)$}
    \BinaryInfC{$\bot$}
    \RightLabel{$1$}
    \UnaryInfC{$\neg B$}
    \TrinaryInfC{$\bot$}
    \UnaryInfC{$B$}
    \RightLabel{$2$}
    \UnaryInfC{$A \to_{i} B$}
     \AxiomC{$\Pi^{*}_{2}$}
     \noLine
    \UnaryInfC{$\neg (A \to _{i} B)$}
    \BinaryInfC{$\bot$}
\end{prooftree}

This is a simplification of the reduction in \cite{Barroso-Nascimento}.

\end{enumerate}

\medskip

\noindent $(\Longleftarrow)$ We show that $t\NEK{[\Gamma]} \vdash_{\NEK} t\NEK[A]$ implies $\Gamma \vdash_{\ECI} A$. Once again we only prove the inductive step for $A \land_{c} B$; the remaining cases are proved in \cite{Barroso-Nascimento}.

\begin{enumerate}

\medskip

    \item The derivation ends with an application of $I \lor_{c}$. Then it has the following shape:

    \begin{prooftree}
        \AxiomC{$[\neg A]$}
        \noLine
        \UnaryInfC{$\Pi_{1}$}
        \noLine
        \UnaryInfC{$\bot$}
        \AxiomC{$[\neg B]$}
        \noLine
        \UnaryInfC{$\Pi_{2}$}
        \noLine
        \UnaryInfC{$\bot$}
        \RightLabel{$I \land_{c}$}
        \BinaryInfC{$A \land_{c} B$}
    \end{prooftree}

The inductive hypothesis yields two deductions $\Pi^{*}_{1}$ and $\Pi^{*}_{2}$. We can construct the following derivation of $(A \land B)^{c}$ in $\ECI$:

\begin{prooftree}
\AxiomC{$[\neg (A \land B)]^{1}$}
\AxiomC{$[A]^{2}$}
\AxiomC{$[B]^{3}$}
\BinaryInfC{$A \land B$}
\BinaryInfC{$\bot$}
\RightLabel{$2$}
\UnaryInfC{$\neg A$}
\noLine
\UnaryInfC{$\Pi^{*}_{1}$}
\noLine
\UnaryInfC{$\bot$}
\RightLabel{$3$}
\UnaryInfC{$\neg B$}
\noLine
\UnaryInfC{$\Pi^{*}_{2}$}
\noLine
\UnaryInfC{$\bot$}
\RightLabel{$1$}
\UnaryInfC{$(A \land B)^{c}$}
\end{prooftree}

\item The derivation ends with an application of $E \land_{c}$. Then it has the following shape:

\begin{prooftree}
   \AxiomC{$\Pi_{1}$}
   \noLine
   \UnaryInfC{$A_{1} \land_{c} A_{2}$}
   \AxiomC{$\Pi_{2}$}
   \noLine
   \UnaryInfC{$\neg A_{j}$}
   \RightLabel{$E_{C}$}
   \BinaryInfC{$\bot$}
\end{prooftree}

The inductive hypothesis yields two deductions $\Pi^{*}_{1}$ and $\Pi^{*}_{2}$. We can construct the following derivation in $\ECI$:

\begin{prooftree}
\AxiomC{$\Pi_{1}^*$}
\noLine
\UnaryInfC{$(A_{1} \wedge A_{2})^{c}$}
\AxiomC{$[A_{1} \wedge A_{2}]^{1}$}
\UnaryInfC{$A_{j}$}
\AxiomC{$\Pi_{2}^*$}
\noLine
\UnaryInfC{$\neg A_{j}$}
\BinaryInfC{$\bot$}
\RightLabel{1}
\UnaryInfC{$\neg (A_{1} \wedge A_{2})$}
\BinaryInfC{$\bot$}
\end{prooftree}

\end{enumerate}

\end{proof}

This means that, in the propositional fragment, $\ECI$ and $\NEK$ are essentially the same logic, especially since $(\bot)^{c}$ and $\bot$ are equivalent \cite[pg. 55]{Barroso-Nascimento}.

\section{Classical universal quantification}


It is usually said that, from an ecumenical perspective, classical logicians and intuitionistic logicians would both recognize themselves in the ecumenical system, in the sense that everything they would like to accept is accepted in the ecumenical system.
 Although true for the intuitionistic logician, obviously this is not completely true in the case of the classical logician; for example, classical implication in $\NE$ system does not satisfy the rule  \textit{modus ponens}: $A , A \to_c B\vdash_{\NE} B$. 
 
In the case of first-order logic, we know we can prove an \textit{ecumenical} result corresponding to  $\neg \forall x\neg A(x) \vdash \exists xA(x)$. But what about $\neg \forall x\ A(x) \vdash \exists x\neg A(x)$? In $\NEK$, 
the introduction and elimination rules for the classical universal quantifier $\forall_{c}$ (see Fig.~\ref{fig:NEK}) 
are clearly \textit{harmonic}:
\begin{prooftree}
\AxiomC{$[\exists_{i} x \neg A(x)]$}
\noLine
\UnaryInfC{$\Pi$}
\noLine
\UnaryInfC{$\bot$}
\RightLabel{$\forall_{c}$-int}
\UnaryInfC{$\forall_{c} xA(x)$}
\RightLabel{$\forall_{c}$-elim}
\UnaryInfC{$\neg \neg A(t/x)$}
\DisplayProof
\qquad
\AxiomC{reduces to}
\DisplayProof
\qquad
\AxiomC{$[\neg A(t/x)]$}
\UnaryInfC{$\exists_{i} x \neg A(x)$}
\noLine
\UnaryInfC{$\Pi$}
\noLine
\UnaryInfC{$\bot$}
\UnaryInfC{$\neg \neg A(t/x)$}
\end{prooftree}

\noindent Moreover, we can easily prove $\neg \forall_{c} x A(x) \vdash_{\NEK} \exists_{c} x\neg A(x)$:

\begin{prooftree}
\AxiomC{$[\exists_{i} x\neg A(x)]^{2}$}
\AxiomC{$[\neg A(a/x)]^{1}$}
\AxiomC{$[\forall_{i} x \neg \neg A(x)]^{3}$}
\UnaryInfC{$\neg \neg A(a/x)$}
\BinaryInfC{$\bot$}
\RightLabel{$1$}
\BinaryInfC{$\bot$}
\RightLabel{$2$}
\UnaryInfC{$\forall_{c} xA(x)$}
\AxiomC{$\neg \forall_{c} x A(x)$}
\BinaryInfC{$\bot$}
\RightLabel{$3$}
\UnaryInfC{$\exists_{c} x\neg A(x)$}
\end{prooftree}

However, as expected, we do not have a Glivenko-type result for the classical universal quantifier $\forall_{c}$: $\vdash_{\NEK} \forall_{c} xA(x)$ does not imply $ \vdash_{\NEK} \neg \neg \forall_{i} xA(x)$ (and $\forall_{c} xA(x) \nvdash_{\NEK} \neg \neg \forall_{i} xA(x)$). It is interesting to observe that we do have a Glivenko-type result for the classical universal quantifier that corresponds to Glivenko's second theorem, $\neg \forall_{c} xA(x) \vdash_{\NEK} \neg \forall_{i} xA(x)$: 
\begin{prooftree}
\AxiomC{$[\exists_{i} x\neg A(x)]^{2}$}
\AxiomC{$[\forall_{i} xA(x)]^{3}$}
\UnaryInfC{$A(a/x)$}
\AxiomC{$[\neg A(a/x)]^{1}$}
\BinaryInfC{$\bot$}
\RightLabel{1}
\BinaryInfC{$\bot$}
\RightLabel{2}
\UnaryInfC{$\forall_{c} xA(x)$}
\AxiomC{$\neg \forall_{c} xA(x)$}
\BinaryInfC{$\bot$}
\RightLabel{3}
\UnaryInfC{$\neg \forall_{i} xA(x)$}
\end{prooftree}

But, as we saw, the Glivenko-type of results are somehow trivial in the system $\ECI$! Even for a universal formula we have that $(\forall xA(x))^{c} \vdash_{\ECI} \neg \neg \forall_{i} xA(x)$ and $\neg (\forall xA(x))^{c} \vdash_{\ECI} \neg \forall_{i} xA(x)$! If we now assume that in the formula $A(x)$ we have no occurrences of the label/constant $c$, we do have something that \textit{looks like} a full Glivenko's first theorem! But we know that Glivenko's first theorem does not extend to universal formulas! What's the trick here? 

In order to understand the real meaning of the Glivenko-type of results we can prove in $\ECI$ we will have a look at some relations between the behavior of the classical operator $\forall_{c}$ and the behavior of the label/constant $c$ applied to a universal formula.

We can emulate an application of the $I_{c}$ rule of $\ECI$ with conclusion $\forall_{c} xA(x)$ by an application of the $\forall_{c}$-Introduction rule in $\NEK$, as well as an application of the $\forall_{c}$-elim rule by an application of the $E_{c}$ rule of $\ECI$.

\begin{theorem} The following hold:
\begin{enumerate}
    \item If $ \neg \forall x A(x) \vdash_{\NEK} \bot$ then $ \vdash_{\NEK} \forall_{c}x A(x)$
    \item If $\vdash_{\ECI} (\forall x A(x))^{c}$ then $\vdash_{\ECI} \neg \neg  A(t/x)$.
\end{enumerate}
\end{theorem}

\begin{proof}
We can construct the following derivation in $\NEK$:
\begin{prooftree}
\AxiomC{$[\exists_{i} x \neg A(x)]^{3}$}
\AxiomC{$[\forall_{i} xA(x)]^{2}$}
\UnaryInfC{$A(a/x)$}
\AxiomC{$[\neg A(a/x)]^{1}$}
\BinaryInfC{$\bot$}
\RightLabel{1}
\BinaryInfC{$\bot$}
\RightLabel{2}
\UnaryInfC{$[\neg \forall_{i} xA(x)]$}
\noLine
\UnaryInfC{$\Pi$}
\noLine
\UnaryInfC{$\bot$}
\RightLabel{3}
\UnaryInfC{$\forall_{c} xA(x)$}
\end{prooftree}
We can also construct the following derivation in $\ECI$:
\begin{prooftree}
\AxiomC{$\Pi$}
\noLine
\UnaryInfC{$(\forall xA(x))^{c}$}
\AxiomC{$[\forall xA(x)]^{1}$}
\UnaryInfC{$A(t/x)$}
\AxiomC{$[\neg A(t/x)]^{2}$}
\BinaryInfC{$\bot$}
\RightLabel{1}
\UnaryInfC{$\neg \forall xA(x)$}
\RightLabel{$E_{c}$}
\BinaryInfC{$\bot$}
\RightLabel{2}
\UnaryInfC{$\neg \neg A(t/x)$}
\end{prooftree}
\end{proof}

But we cannot emulate the rule $E_{c}$ of $\ECI$ by means of the rule $\forall_{c}$-elim, and the rule  $\forall_{c}$-int by means of the rule $I_{c}$ of $\ECI$, and this means that the deductive behavior of the classical operator is different from the deductive behavior of the labeled formula. In a certain sense, a formula $\forall_{c} xA(x)$ can be \textit{interpreted} as $\neg \exists  x\neg A(x)$, whereas a formula $(\forall xA(x))^{c}$ can be \textit{interpreted} as $\neg \neg \forall xA(x)$ (see the translation $t\ECI$ in Definition \ref{TECI}), and these \textit{interpretations} are not intuitionistically equivalent! This peculiar behaviour is only observed in the universal quantifier, which is entirely expected because Glivenko's theorems hold for the $\forall$-free fragments of $\CL$ and $\IL$. And now it is possible to explain in which sense the Glivenko-type results are trivial in $\ECI$: what the theorem $A^{c} \vdash_{\ECI} \neg \neg A^{i}$ says can be interpreted simply as 
$\neg \neg A \vdash_{\ECI} \neg \neg A$, and in the particular case of universal formulas, as $\neg \neg \forall xA(x) \vdash_{\ECI} \neg \neg \forall xA(x)$. Mystery solved!

\section{Some conceptual remarks concerning the relation between $\ECI$ and $\NEK$}

Natural deduction allows us to fix the meaning of a logical connective by specifying the rules governing its use. The non-interdefinability of intuitionistic operators makes it so that intuitionistic specifications are expected to be independent of each other, but the same does not hold for classical specifications due to the similarity of grounds for classical use. Consequently, classical logic can be obtained by adding to intuitionistic logic autonomous rules permitting the use of classical proof principles (such as the \textit{classical reductio}), which modify the meaning of connectives by uniformly supplying them with indirect (classical) means of proof. As such, to obtain classical logic from intuitionistic logic it suffices to change the notion of proof by adding a rule which allows classical reasoning.

From a different perspective, it could be argued that the possibility of defining an autonomous classical rule is a byproduct of the uniformity of changes in the meaning of connectives, but  this does not imply the existence of a change in the concept of proof. The classical grounds for use must be included in the individual definition of each connective, but since they must be included in every connective it is also possible to implement this through the definition of a single autonomous rule. This means that classical proof rules are merely technical tools for changing the definition of all connectives at once, but that a conceptually faithful classical definition would have to include classical grounds for use directly into each of the introduction and elimination rules for operators instead.

The differences between $\ECI$-type systems and Prawitz-type systems (which includes $\NE$ and $\NEK$) seem to be explained by the differences between both perspectives. In the first one, just like in $\ECI$, the change operates at the level of proofs, so classical connectives are obtained by equipping intuitionistic logic with classical means of proof that indirectly change the meaning of connectives when used. In the second one, just like in $\NE$ and $\NEK$, the changes are made directly at the level of connectives, so we only have one notion of proof but are now allowed to use it together with connectives that are explicitly defined in terms of classical grounds for use. The difference is subtle but, as our study shows, not without consequence. In particular, the principles of each path lead us to distinct versions of logical ecumenism.

In a certain sense, we can summarize the differences between both approaches in the following way: while Prawitz's and Krauss' systems have actual \textit{classical operators}, the system $\ECI$ distinguish by means of the constant $c$ a classical behavior from an intuitionistic behavior of the same operator (remember that in $\ECI$ we have just one set of logical operators and formulas can be labeled with the \textit{constant} $c$ to indicate this \textit{classical behavior}). If we restrict the two approaches to the propositional or $\forall$-free fragment, it is indifferent whether we use $\ECI$ or $\NEK$. But this is not true when we add the universal quantifier, and the question now of which approach corresponds more faithfully to a classical universal quantifier is everything but negligible.

\section*{Acknowledgments}
First of all, we would like to thank Marcelo Coniglio for being such an inspiration and a good friend.

Barroso-Nascimento was supported in part by the Coordena\c c\~ao de Aperfei\c coamento de Pessoal de N\'ivel Superior - Brasil (CAPES) - Finance Code 001. Pereira is supported by the following projects: CAPES/COFECUB 88881.878969/2023-01, CNPq-313400/2021-0, and CNPq-Gaps and Gluts. Pimentel has received funding from the European Union's Horizon 2020 research and innovation programme under the Marie Sk\l odowska-Curie Grant agreement Number 101007627. Pimentel and Barroso-Nascimento are supported by the Leverhulme Trust grant RPG-2024-196. 

This work has benefitted from \href{https://www.dagstuhl.de/24341}{Dagstuhl Seminar 24341} ``Proof Representations: From Theory to Applications.''

The authors are grateful for the useful suggestions from the anonymous referee.


\end{document}